\documentclass{llncs}

\usepackage{simplemargins}
\usepackage{times,amsmath,amssymb,xspace}
\usepackage{txfonts}
\usepackage{color}
\usepackage{graphicx}
\usepackage{multirow}
\usepackage{amsmath}
\usepackage{amssymb}
\usepackage{theorem}
\usepackage{epsfig}
\usepackage{fancyvrb}
\usepackage{url}
\usepackage{cite}
\usepackage{multirow}
\usepackage{wrapfig}
\usepackage{comment}
\usepackage{authblk}
\usepackage{array}
\usepackage{enumitem}

\setleftmargin{1.3in}
\setrightmargin{1.3in}
\settopmargin{1.3in}
\setbottommargin{1.3in}

\numberwithin{equation}{section}	
\theoremstyle{plain}			
\newtheorem{thm}{Theorem}
\newtheorem{Theorem}[thm]{Theorem}
\newtheorem{Lemma}[thm]{Lemma}

\theoremstyle{definition}		
\newtheorem{Definition}[thm]{Definition}

\newcommand{\vsep}{\hspace{3mm} | \hspace{3mm}}

\begin{document}
\title{Undecidability of a Theory of Strings,\\ Linear Arithmetic over
  Length, and String-Number Conversion}

\author{Vijay Ganesh and Murphy Berzish}
\institute{University of Waterloo\\
\{vganesh, mtrberzi\}@uwaterloo.ca\\
}

\maketitle

\begin{abstract}
\label{sec:abstract}
In recent years there has been considerable interest in theories over
string equations, length function, and string-number conversion
predicate within the formal verification and computer security
communities. SMT solvers for these theories, such as Z3str2, CVC4, and
S3, are of immense practical value in exposing security
vulnerabilities in string-intensive programs. At the same time, these
theories are of great interest to logicians, with many open
questions relating to their decidability and complexity.

\vspace{0.1cm} Motivated by these open questions and above-mentioned
applications, we study a first-order, many-sorted, quantifier-free
theory $T_{s,n}$ of string equations, linear arithmetic over string
length, and string-number conversion predicate and prove three
theorems. First, we prove that the satisfiability problem for the
theory $T_{s,n}$ is undecidable via a reduction from a theory of
linear arithmetic over natural numbers with \textit{power predicate},
we call power arithmetic. Second, we show that the string-numeric
conversion predicate is expressible in terms of the power predicate,
string equations, and length function. This second theorem (in
conjunction with the reduction we propose for the undecidability
theorem) suggests that the power predicate is expressible in terms of
word equations and length function if and only if the string-numeric
conversion predicate is also expressible in the same fragment. Such
results are very useful tools in comparing the expressive power of
different theories, and for establishing decidability and complexity
results. Third, we provide a consistent axiomatization $\Gamma$ for
the functions and predicates of $T_{s,n}$, and prove that the theory
$T_{\Gamma}$, obtained via logical closure of $\Gamma$, is not a
complete theory.
\end{abstract}

\section{\bf Introduction} 

The satisfiability problem for theories over finite-length strings
(aka words) has long been studied by mathematicians such as
Quine~\cite{Quine}, Post, Markov and Matiyasevich~\cite{Matiyasevich},
Makanin~\cite{makanin}, and Plandowski~\cite{KarhumakiPM97,
  plandowski99, plandowski2006}. Post, Markov, and Quine were
motivated by the connections between theories over word
equations\footnote{In this paper, we interchangeably use the terms
  \textit{word equations} and \textit{string equations}. The term
  ``word equations'' is the convention among logicians, while formal
  verification researchers tend to use the term ``string equations''.}
and Peano arithmetic, while Matiyasevich's motivation for studying
them was their connection to Diophantine
equations~\cite{Matiyasevich}.

More recently there has been considerable interest in efficient
solvers for theories over string equations in the formal verification,
software engineering, and security research communities.  Examples of
such solvers include Z3str2~\cite{z3strcav} and CVC4~\cite{Liang2014},
both of which support the quantifier-free (QF) first-order many-sorted
theory $T_{s,n}$ of string equations, length, and string-integer
conversions.  This theory is expressive enough that many string-related
library functions and programming constructs from languages such as C,
C++, Java, PHP, and JavaScript can be easily encoded in terms of its
functions and predicates. The expressive power of $T_{s,n}$ and
efficient practical string solvers have enabled many applications in
program analysis and verification~\cite{z3strcav, prateek,
  rupak}. Examples include dynamic symbolic execution aimed at
automated bug-finding~\cite{hampi, prateek}, and analysis of
database/web applications~\cite{emmiMS2007, rupak, WassermannSu2007}.

Given the fundamental nature of the theory $T_{s,n}$ and its fragments
(e.g., note that word equations essentially form a free semigroup
studied intensively by mathematicians over the last several
decades~\cite{lothaire}), it is no surprise that there is strong
motivation from logicians to study their decidability and
complexity. In the 1940's, Post and Markov conjectured that the
fully-quantified first-order theory of string equations (i.e.,
quantified sentences over Boolean combination of string equations)
must be undecidable. In his 1946 paper, Quine \cite{Quine} showed that
this theory is indeed undecidable. In 1977, Makanin famously proved
that the satisfiability problem for the quantifier-free theory of
string equations is decidable~\cite{makanin}. This result is often
considered as one of the most complex proofs in theoretical computer
science. In recent years, Plandowski and others considerably improved
Makanin's results and showed that satisfiability problem for string
equations is in PSPACE~\cite{plandowski2006}. In 2012, Ganesh et
al. showed that $\forall\exists$-fragment of positive string equations
is undecidable, strengthening Quine's result and
establishing the boundary between decidability and undecidability for
string equations~\cite{ganesh2012}. Additionally, Ganesh et al. also
proved conditional decidability results for the quantifier-free theory
of string equations and linear arithmetic over the string length
function~\cite{ganesh2012}.

As automated reasoning tools and algorithms for the satisfiability
problem for the theory $T_{s,n}$ continue to be intensively researched
and developed, it is a natural question to ask whether the theory is
indeed decidable. This question has been open for at least over the 15
years since interest in string solvers dramatically increased in the
formal methods community, and is the primary focus of this paper.

\subsection{Problem Statement}

We answer the following three questions in this paper:

\begin{enumerate}
\item Is the satisfiability problem for the quantifier-free fragment
  of a first-order two-sorted theory $T_{s,n}$ of finite-length
  strings over a finite alphabet decidable, whose functions and
  predicates are as follows: concatenation function and the equality
  predicate over string terms, string to natural number conversion
  predicate, length function from string terms to natural numbers, and
  linear arithmetic over natural numbers and length function.

  An answer to this question may give us clues to decidability
  questions relating to certain fragments of $T_{s,n}$ that remain
  open. For example, it is not known whether the quantifier-free
  theory of word equations and equality over the length function is
  decidable, and this problem has been open for at least 5
  decades~\cite{Matiyasevich}. Furthermore, as discussed above, the answer
  impacts practical string solvers such as
  CVC4 and Z3str2 which in currently implement incomplete algorithms
  to decide the satisfiability problem for the theory $T_{s,n}$. We
  show that the satisfiability problem for the theory $T_{s,n}$ is
  undecidable.

\item Is the string-numeric conversion predicate expressible in terms
  of string equations and length function? This question is important
  from a theoretical point of view because if such an expressibility result exists,
  then this immediately settles the open question
  regarding the satisfiability problem for the quantifier-free theory
  of string equations and length.
  
\item What is a consistent (possibly minimal) axiomatization $\Gamma$
  for the functions and predicates of $T_{s,n}$? Is the first-order
  many-sorted fully-quantified theory $T_{\Gamma}$ obtained as a
  logical closure of the axiom set $\Gamma$ complete? (Note that the
  existential closure of the quantifier-free first-order many-sorted
  theory $T_{s,n}$ is a subset of $T_{\Gamma}$.)
\end{enumerate}


\subsection{Contributions in Detail}

In greater detail, the contributions of this paper are as follows:

\begin{enumerate}
\item We prove that the satisfiability problem for the quantifier-free
  theory of string equations, linear arithmetic over string length,
  and string-number conversion is undecidable. This problem has been
  open for some time, and is of great interest to formal verification
  researchers~\footnote{Note that the theory $T_{s,n}$ is stronger
    than the quantifier-free theory of string equations and linear
    arithmetic over string length function, since $T_{s,n}$
    additionally has the string-number conversion predicate.}. The
  ability to model string concatenation, equality, linear arithmetic
  over length, and string-natural number conversions is particularly
  useful in identifying security vulnerabilities in applications
  developed using many modern programming languages, including
  JavaScript web
  applications~\cite{OISSTA14,saxena10kudzu}. (Section~\ref{sec:undecidability})

\item We also show that the $\pi$ predicate from the power arithmetic
  theory $T_p$, which asserts the equality $z = x * 2^y$,
  is expressible in terms of the $numstr$ predicate from
  $T_{s,n}$. More precisely, we encode $\pi$ using only the $numstr$
  predicate, string equations, and string length function. In the
  above-mentioned undecidability theorem, we establish that
  $numstr$ can be encoded using only the $\pi$ predicate, string
  equations, and string length function. These two reductions put
  together suggest that the $\pi$ predicate is expressible using string
  equations and length function iff $numstr$ is. Expressibility results
  are very useful tools in constructing reductions, distinguishing the
  expressive powers of various theories, and in establishing
  (un)-decidability results. Additionally, our expressibility results
  suggest that the $numstr$ predicate is much more complex, both from
  a theoretical and a practical point-of-view, than it seems at first
  glance. (Section~\ref{sec:definability})

\item We establish a consistent finite axiomatization $\Gamma$ for the
  functions and predicates in the language $L$ of
  $T_{s,n}$. Additionally, we show that the first-order many-sorted
  $L$-theory $T_{\Gamma}$, that is the closure of the axioms $\Gamma$,
  is not a complete theory. That is, there are $L$-sentences $\phi$
  such that $T_{\Gamma}$ does not entail either $\phi$ or its
  negation. (Section~\ref{sec:soundness})
\end{enumerate}

The paper is organized as follows: In Section~\ref{sec:prelim} we
provide the syntax and semantics of the theory $T_{s,n}$. In
Section~\ref{sec:undecidability} we prove the undecidability of the
satisfiability problem of $T_{s,n}$.  In
Section~\ref{sec:definability}, we show a reduction from the power
arithmetic theory to $T_{s,n}$. In Section~\ref{sec:soundness}, we
discuss the consistency of an axiom system $\Gamma$ for the language of $T_{s,n}$, and
in Section~\ref{sec:incompleteness} we establish
that the theory $T_{\Gamma}$ is incomplete. In
Section~\ref{sec:relwork} we provide a comprehensive overview of the
decidability/undecidability results for theories of strings over the
last several decades, and the practical relevance of this theory
in the context of verification and security. Finally, we conclude in
Section~\ref{sec:conclusions}, and provide a list of open problems
related to various extensions and fragments of the theory $T_{s,n}$
some of which have been open for many decades now.

\section{Preliminaries}
\label{sec:prelim}

In this section, we define the syntax and semantics of the
first-order, many-sorted, language $L$ of string (aka word) equations
with concatenation, length function over string terms, linear
arithmetic over natural numbers and the length function, and
string-number conversion predicate.
In Section~\ref{sec:soundness}, we will present an axiom system $\Gamma$ for this language
and prove that it is consistent.

\subsection{The Language $L$: Syntax for Theories over String Equations, Length, and String-Number Conversion}

We first define the countable language $L$ below, i.e., its sorts, and
constant, function, and predicate symbols.

\begin{enumerate}

\item {\bf Sorts:} The language is many-sorted, with a string sort
  $str$ and a natural number sort $num$. The Boolean sort $Bool$ is
  standard. When necessary, we write the sort of an $L$-term $t$
  explicitly as $t:sort$.

\item {\bf Finite Alphabet:} We fix a finite alphabet $\Sigma =
  \{0,1\}$ over which all strings are defined. As necessary, we may
  subscript characters of $\Sigma$ with an $s$ to indicate that their
  sort is str.

\item {\bf String and Natural Number Constants:} We fix a two-sorted
  set of constants $Con = Con_{str} \cup Con_{num}$.  The set
  $Con_{str}$ is a subset of $\Sigma^*$, the set of all finite-length
  string constants over the finite alphabet $\Sigma$. Elements of
  $Con_{str}$ will be referred to as {\it string constants} or simply
  {\it strings}. The empty string is represented by
  $\epsilon$. Elements of $Con_{num}$ are the {\emph natural numbers}
  starting from 0. As necessary, we may subscript numbers by $n$ to
  indicate that their sort is num.

\item {\bf String and Numeric Variables:} We fix a disjoint two-sorted
  set of variables $var = var_{str} \cup var_{num}$; $var_{str}$
  consists of string variables, denoted $X,Y,S, \ldots$ that range
  over string constants, and $var_{num}$ consists of numeric
  variables, denoted $m,n,\ldots$ that range over the natural numbers.

\item {\bf String Function Symbols:} The string function symbols
  include the concatenation operator $\cdot: str \times str \rightarrow str$
  that take as argument two string terms and outputs a string term,
  and the length function $len: str \rightarrow num$ that takes as
  argument a string and outputs a natural number.

\item {\bf Linear Arithmetic Function Symbols:} The natural number
  (aka numeric) function symbols include the addition symbol $+: num \times
  num \rightarrow num$, that takes as argument two numeric terms and
  outputs a numeric term. (Following standard practice in mathematical
  logic literature, we allow multiplication by constants as a
  shorthand.)

\item {\bf String Predicate Symbols:} The predicate symbols over
  string terms include the equality symbol $=_s: str \times str \rightarrow
  Bool$ that takes as argument two string terms and evaluates to a
  Boolean value, and the string-number conversion predicate
  $numstr:num \times string \rightarrow Bool$.

\item {\bf Natural Number Predicate Symbols:} The predicate symbols
  over natural number terms include the equality symbol $=_n: num \times
  num \rightarrow Bool$, and the inequality predicate $\leq: num \times num
  \rightarrow Bool$.

\end{enumerate}

\subsection{Terms and Formulas in the Language $L$}

\noindent{\bf Terms:} $L$-terms may be of string or numeric sort. A
string term ($t_{str}$ in Figure~\ref{fig:syntax}) is inductively
defined as either an element of $var_{str}$, an element of
$Con_{str}$, or a concatenation of string terms (denoted by the
function $concat$ or interchangeably by the $\cdot$ operator).  A numeric or
natural number term ($t_{num}$ in Figure~\ref{fig:syntax}) is an
element of $var_{num}$, an element of $Con_{num}$, the length function
applied to a string term, a constant multiple of a length term, or a
sum of length terms. (Note that for convenience we may write
concatenation and addition as $n$-ary functions, even though we define
them as binary operators.)

\noindent{\bf Atomic Formulas:} There are five types of atomic
formulas as given in Figure~\ref{fig:syntax}: (1) word equations
($A_{wordeqn}$), (2) linear arithmetic predicates over natural numbers
and length constraints ($A_{num}$), and (3) string-numeric conversion
predicates ($A_{numstr}$).

\noindent{\bf Quantifier-free Formulas:} Boolean combination of atomic
formulas. The term ``quantifier-free'' formulas means
that each free variable is implicitly existentially quantified and no
explicit quantifiers may be written in the formula.

\noindent{\bf Formulas and Prenex-normal Form:} $L$-Formulas are
defined inductively over atomic formulas (see
Figure~\ref{fig:syntax}). The symbol Qx refers to a block of
quantifiers over a set $x$ of variables. We assume that formulas are
always represented in prenex-normal form (a block of quantifiers
followed by a quantifier-free formula).

\noindent{\bf Free and Bound Variables, and Sentences:} We say that a
variable under a quantifier in a formula $\phi$ is bound. Otherwise we
refer to variables as free. A formula with no free variables is called
a sentence.

\begin{figure}[t!]
\[
\begin{array}{llll}
  F & \Coloneqq & Atomic  \vsep F \wedge F  \vsep F \vee F  \vsep \neg F \vsep Qx. F\\
  Atomic & \Coloneqq & A_{wordeqn}  \vsep A_{num}  \vsep A_{numstr} & \\
  A_{wordeqn} & \Coloneqq & t_{str} = t_{str} & \\
  A_{num} & \Coloneqq & t_{num} = t_{num} \vsep t_{num} < t_{num} &  \\
  A_{numstr} & \Coloneqq & numstr(n, s) \\
  & & \text{where } n \in t_{num}, s \in t_{str} \\
  t_{str} & \Coloneqq & a \vsep X \vsep concat(t_{str},...,t_{str}) \\
  & & \text{where} \hspace{1mm} a \in Con_{str} \hspace{1mm} \& \hspace{1mm} X \in var_{str}\\
  t_{num} & \Coloneqq & m \vsep v \vsep len(t_{str}) \vsep t_{num} + t_{num} \\
  & & \text{where} \hspace{1mm} m \in Con_{num} \hspace{1mm} \& \hspace{1mm} v \in var_{num}\\
\end{array}
\]
\caption{\label{fig:syntax} The syntax of $L$-formulas.}
\end{figure}

\subsection{Signature of the Theory $T_{s,n}$}

We define the signature of $T_{s,n} = \left< \Sigma^{*}, \mathbb{N},
0_s, 1_s, \cdot, 0_n, 1_n, +, len, numstr, =_{s}, =_{n}, <_{n}
\right>$, where $\Sigma^{*}$ is the set of all string constants over a
finite alphabet $\Sigma$, $\mathbb{N}$ is the set of natural numbers,
$\cdot$ is the two-operand string concatenation function, $+$ is the
two-operand addition function for natural numbers, $len$ is a
function that takes a string and returns its length as a natural
number, $=_{s}$ is the equality predicate over strings, $=_{n}$ and
$<_{n}$ are the equality and less-than predicates over natural
numbers, and $numstr$ is a two-argument predicate such that $numstr(i,
s)$ is true for natural number $i$ and string $s$ if and only if $s$
is a valid binary representation of the natural number $i$.  By a
``valid binary representation'' we mean that $s$ does not contain any
characters other than `0' and `1', and interpreting the characters of
$s$ as the digits of a numeral in base 2, where the last character of
$s$ is the least significant digit, produces a natural number that is
equal to $i$.  (Hence we require that the alphabet $\Sigma$ contain
characters `0' and `1'.)
Note that the signatures of all theories
considered in this paper are countable.

\subsection{$L$-Semantics and the Canonical Model $\mathbb{M}$}\
\label{sec:canonicalmodel}

In this section, we provide semantics for the symbols in the language
$L$ via what we call a canonical model $\mathbb{M}$. We take the
finite alphabet $\Sigma$ to be the set $\{0,1\}$. The results
presented here can be easily extended to other finite alphabets. We
assume standard definitions for the terms \textit{interpretation of
  symbols} and \textit{model}~\cite{HodgesModelTheory}.

\noindent{\bf Universe of Discourse for symbols in $L$:} The universe
of discourse over which all symbols are interpreted is two-sorted
disjoint sets. The first set $\Sigma^*$, of sort str, is the set of all
finite-length strings over the alphabet $\Sigma = \{0,1\}$ including
the empty string (represented by $\epsilon$), and the second set $\mathbb{N}$, of sort
num, is the set of natural numbers starting from $0$.

\noindent{\bf Interpretation of Natural Number Variables, Constants,
  Functions and Predicates:} Variables of num sort range over the set
$\mathbb{N}$ of natural numbers, and constants represent corresponding
natural numbers. Note that all natural number constants are
represented as binary numbers, unless otherwise specified. The
function $+$ and the predicates $=_n, \leq$ have the standard
interpretations. (Multiplication by constant is also treated in the
standard way as a shorthand for repeated addition.)

\noindent{\bf Interpretation of String Variables, Constants,
  Functions, and Predicates:} String constants are interpreted as
a finite concatenation of letters $0$ and $1$ and correspond to appropriate
strings in $\Sigma^*$, and string variables range over values from
$\Sigma^*$. The string concatenation function is inductively defined
over elements of $\Sigma^*$ in the natural way.

\noindent{\bf What is meant by the Length of a String:} For a string
or a word, $w$, $len(w)$ denotes the length of $w$, or equivalently,
the (natural) number of characters from $\Sigma$ in the interpretation
of $w$ under a given assignment.

\noindent{\bf The Meaning of $numstr$ Predicate:} The $numstr$
predicate asserts that the interpretation of its string argument is a
valid binary representation of the natural number represented by its
numeric argument.  A string $s$ is a valid binary representation of a
natural number $i$ iff the following properties hold:

\begin{enumerate}
\item $s$ does not contain any characters in $\Sigma$ other than `0'
  and `1'.
\item Let $s[n]$ denote the $n$th character in $s$, where $n$ is a
  natural number between 0 and $len(s) - 1$ inclusive.  Let $s'[n]$
  denote the numeric value of $s[n]$, where $s'[n] = 1$ if $s[n]$ is
  `1', and $s'[n] = 0$ if $s[n]$ is `0'.  Then it must be the case
  that $\sum_{n=0}^{len(s) - 1} s'[n] 2^{len(s) - n - 1} = i$.
  (Here we expand the characters of $s$ into a binary representation
  of $i$.)
\end{enumerate}

\noindent{\bf The Meaning of Equality between String Terms:} For a
word equation of the form $t_1 = t_2$, we refer to $t_1$ as the left
hand side (LHS), and $t_2$ as the right hand side (RHS). Two string
terms are considered equal if their interpretations have the same
characters appearing in the same order, i.e., the LHS and RHS evaluate
to the same string in $\Sigma^*$ under the appropriate interpretation
for variables and constants in the LHS and RHS of the given equality.

\noindent{\bf The Canonical Model:} This interpretation
of $L$-symbols along with the universe of discourse defines the
canonical $L$-model. (An interpretation of a set of symbols in a
language $L$ along with universe of discourse is called an $L$-model.)

\subsection{Standard Logic Definitions}

Here we give some standard definitions such as assignment,
satisfiability, validity, consistency of an axiom system, and
completeness of a theory.

\noindent{\bf Assignments, Satisfiability, Validity, and Equisatisfiability:} Given an
$L$-formula $\theta$, an {\it assignment} for $\theta$ (with respect
to $\Sigma)$ is a map from the set of free variables in $\theta$ to
$\Sigma^* \cup \mathbb{N}$ (where string variables are mapped to
strings and natural number variables are mapped to numbers).  Given
such an assignment, $\theta$ can be interpreted as an assertion about
$\Sigma^*$ and $\mathbb{N}$.  If this assertion is true, then we say
that $\theta$ itself is {\it true} under the assignment.  If there is
some assignment which makes $\theta$ true, then $\theta$ is called
{\it satisfiable}.  An $L$-formula with no satisfying assignment is
called an {\it unsatisfiable} formula. We say two formulas $\theta, \phi$ are {\it equisatisfiable} if $\theta$ is satisfiable iff $\phi$
is satisfiable.  Note that this is a broad definition: equisatisfiable
formulas may have different numbers of assignments and, in fact, need
not even be from the same language. We say a formula is valid if it is
true under all possible assignments.

\noindent{\bf The Satisfiability Problem:} The {\it satisfiability
  problem} for a set $S$ of formulas is the problem of deciding
whether any given formula in $S$ is satisfiable or not. We say that
the satisfiability problem for a set $S$ of formulas is decidable if
there exists an algorithm (or \textit{satisfiability procedure}) that
solves its satisfiability problem. Satisfiability procedures must have
three properties: soundness, completeness, and termination. Soundness
and completeness guarantee that the procedure returns ``satisfiable"
if and only if the input formula is indeed satisfiable. Termination
means that the procedure halts on all inputs. In a practical
implementation, some of these requirements may be relaxed for the sake
of improved typical performance. Analogous to the definition of the
satisfiability problem for formulas, we can define the notion of the
\textit{validity problem} (aka decision problem) for a set $Q$ of
sentences in a language $L$. The validity problem for a set $Q$ of
sentences is the problem of determining whether a given sentence in
$Q$ is true under all assignments.

\noindent{\bf Logical Entailment:} We say that a set of sentences $C$
entails a sentence $\phi$, written as $C \models \phi$, if any
model $A$ of $C$ is also a model of $\phi$. We say a model $A$ is a
model of a set of sentences $C$, if all sentences of $C$ are true under
some assignments in $A$, written as $A \models C$.

\noindent{\bf Consistency of an Axiom System:} A set of $L$-sentences
may be designated as axioms. We say that an axiom system $A$ is
consistent if for any $L$-formula $\phi$, the axiom system $A$ does
not logically imply both a formula $\phi$ and its negation $\neg \phi$.

\noindent{\bf Theory, Closure of an Axiom System, Completeness of a
  Theory:} A set of $L$-sentences is referred to as a theory. The
closure $C$ of an axiom system $A$ is the set of sentences that are
logically implied by $A$, i.e., every model of $A$ is a model of the
set $C$. We say that a theory $T$ is complete if for every
$L$-sentence $\phi$, $T$ logically entails either $\phi$ or its
negation.

\section{Result 1: The Undecidability of the Satisfiability Problem for $T_{s,n}$}
\label{sec:undecidability}

In this section we prove that the satisfiability problem for the
first-order many-sorted quantifier-free theory $T_{s,n}$ over string
equations and linear arithmetic over natural numbers extended with
string length and a string-number conversion predicate is undecidable.

\subsection{The Theory of Power Arithmetic $T_{p}$, and B\"{u}chi's Results}

In this subsection, we present the syntax and semantics of the power
arithmetic theory $T_{p}$, and discuss B\"{u}chi's results for this
theory.

\subsubsection{Syntax, Semantics, and the Signature of Theory $T_{p}$}

We define the theory $T_{p}$ to have the signature $\left< \mathbb{N}, 0, 1, +, \pi,
<_{n}, =_{n} \right>$, where $\mathbb{N}$ is the set of natural
numbers, 0 and 1 are constants, $+$ is the two-operand addition
function, $<_{n}$ and $=_{n}$ are the two-operand less-than and
equality predicates, and $\pi$ is a three-operand predicate
\footnote{Representation of $\pi$ as a predicate is somewhat more
  natural given that string-number conversion is also represented as a
  predicate.}  defined as $\pi(p, x, y) \iff p = x \times 2^{y}$. Note
that we only consider the satisfiability problem over the
quantifier-free fragment of $T_p$ (equivalently the existential
closure over quantifier-free formulas).

\subsubsection{B\"{u}chi's Undecidability Result}
\label{sec:buchiresult}

Below we briefly present the necessary context for B\"{u}chi's
undecidability result for theory $T_p$.
We note that Lemmas~\ref{lem:robinson} and~\ref{lem:buchi1},
as well as the statement of Theorem~\ref{thm:TpUndecidable},
are adapted from~\cite{Buchi90} where they were originally presented.

\begin{Lemma}{\bf (Julia Robinson's divisibility lemma)} \label{lem:robinson}
  If $m \le n, l > 2n^2,$ and $l+m, l-m | l^2 - n$, then $m^2 = n$.
  (Refer to Lemma~5 in~\cite{Buchi90}.)
\end{Lemma}

\begin{Lemma}{\bf (B\"{u}chi's Lemma)} \label{lem:buchi1}
  In $T_p = \left< \mathbb{N}, 0, 1, +, \pi \right>$ we can
  existentially define addition and multiplication
  on $\mathbb{N}$.  (Refer to Lemma~6 in~\cite{Buchi90}.)
\end{Lemma}

\begin{Theorem}{\bf (B\"{u}chi's Undecidability Theorem)} \label{thm:TpUndecidable}
  The existential theory of $T_p = \left< \mathbb{N}, 0, 1, +, \pi
  \right>$ is undecidable.  (Corollary~5 in~\cite{Buchi90}.)
\end{Theorem}


\subsection{Proof Idea}

We present a sound, complete, and terminating (recursive) reduction
from the satisfiability problem for the theory of power arithmetic,
$T_{p}$, which is an extension of arithmetic over natural numbers with
a three-argument $\pi$ predicate defined as $\pi(p, x, y) \iff p = x *
2^{y}$, to the satisfiability problem of the theory $T_{s,n}$. This
theory $T_p$ (and its associated satisfiability problem) was shown by
B\"{u}chi to be undecidable~(in \cite{Buchi90}, as outlined in
Section~\ref{sec:buchiresult}).

As the theory $T_{s,n}$ already has arithmetic over natural numbers,
the only detail that is missing is an encoding of the $\pi$ predicate
into $T_{s,n}$. Recall that in
bit-vector arithmetic, an unsigned left shift corresponds to
multiplication by a power of 2.  Therefore, if we have a binary string
that represents the natural number $x$ and we concatenate this string
with a string of all zeroes of a given length $y$, the resulting
string will be the binary representation of $x * 2^{y}$. Once this
encoding is provided, then it is easy to see that any quantifier-free
formula in $T_p$ can be reduced equisatisfiably to a quantifier-free
formula in $T_{s,n}$.

\subsection{The Undecidability Theorem}

\begin{Theorem}
  The satisfiability problem for the theory $T_{s,n}$ is undecidable.
\end{Theorem}

\begin{proof}

  We prove this result via a recursive reduction from the theory $T_p$
  (B\"{u}chi's power arithmetic) to theory $T_{s,n}$, i.e., any
  quantifier-free formula in $T_p$ can be equisatisfiably reduced to a
  quantifier-free formula in $T_{s,n}$. Thus, if the satisfiability
  problem for $T_{s,n}$ is decidable then so is the satisfiability
  problem for $T_p$. By B\"{u}chi's theorem~\cite{Buchi90} the
  satisfiability problem for $T_p$ is undecidable, and hence so is the
  satisfiability problem for $T_{s,n}$.

\noindent{\bf The Reduction from $T_p$ to $T_{s,n}$.} We reduce each
constant, function, predicate, and atomic formula of $T_{p}$ to
$T_{s,n}$ by applying the following rules recursively over the input formula:

\begin{enumerate}

\item Each natural number in $\mathbb{N}$ is represented directly as a
  constant in $T_{s,n}$.

\item Variables in $T_{p}$ are represented directly as variables of
  numeric sort in $T_{s,n}$.

\item Addition of two terms $t_1 + t_2$ is represented directly as
  addition over natural numbers, $t_1 + t_2$, in $T_{s,n}$.

\item Equality of terms in $T_{p}$ is represented directly via a
  recursive reduction as equality $t_1 =_{n} t_2$ of terms of numeric
  sort.

\item The less-than predicate in $T_{p}$ is represented directly as
  comparison of natural numbers, $t_1 <_{n} t_2$.

\item The predicate $\pi(p, x, y)$ is expressible as follows: $\exists
  z:str, \exists x_{s}:str : \, (``0" \cdot z = z \cdot ``0" \land
  len(z) = y \land numstr(p, x_{s} \cdot z) \land numstr(x, x_{s})
  )$. The interpretation of the $\pi$ predicate is $p = x\times
  2^y$. The variables $z$ and $x_{s}$ are string variables, and $z$ is
  a string of the ``0'' character of length equal to y. The $x_s$
  variable is the string binary representation of the natural number
  $x$. The concatenation of $x_s$ followed by $z$ is a binary
  representation of $p$. It is easy to verify that the given formula
  over free numeric variables $x,y,p$ is satisfiable iff $\pi(p,x,y)$
  is satisfiable.
\end{enumerate}

\noindent{The reduction can easily be extended to arbitrary
  quantifier-free formulas in $T_p$. It is easy to verify that the
  reduction is sound, complete, and terminating for all
  inputs.}
\end{proof}

\subsection{Discussion}

Recall that the satisfiability problem for the theory of
quantifier free string equations with string length remains
open. Knowing whether that theory is decidable would be of value in
many program analysis applications. The theory $T_{s,n}$ we consider
here is arguably more directly relevant to program analysis since many
state-of-art solvers implement exactly this theory, as the extension
of string-number conversion allows it to model similar operations
which are present in almost all programming languages that have a data
structure for strings. Examples of programming language
operations/functions that could be modelled with the string-numeric
conversion predicate include JavaScript's \texttt{parseInt} and
\texttt{toNumber} methods, which perform integer-string and
string-integer conversion.

\section{Result 2: Expressibility}
\label{sec:definability}

In this section we establish that the $\pi(p,x,y)$ and $numstr$
predicates are expressible in terms of each other. 
We define a new theory $T_\pi$ (different from $T_p$), which is the same
as $T_{s,n}$ except that $numstr$ is removed and replaced by the
$\pi(p,x,y)$ predicate. From the previous section it is clear that
any formula involving the $\pi(p,x,y)$ predicate can be reduced to some
formula in $T_{s,n}$ using some Boolean combination of $numstr$
predicate, string equations, and length function.
This shows us that a reduction exists from $T_\pi$ to $T_{s,n}$.  
We now show that a reduction in the opposite direction exists; that is,
the $numstr$ predicate can be expressed in terms of quantified
formulas over the $\pi(p,x,y)$ predicate, word equations, and length
function.

The value of these two recursive reductions is that it suggests that
the $\pi$ predicate is expressible using string equations and length
function iff $numstr$ is. Expressibility results are very useful tools
in constructing reductions, distinguishing the expressive powers of
various theories, and establishing (un)-decidability
results. Additionally, our expressibility results suggest that the
$numstr$ predicate is much more complex, both from a theoretical and a
practical point of view, than it seems at first glance.

\begin{Definition}
  A predicate $P$ is \textbf{expressible} in some theory $T$ having
  language $L_{T}$ if there exists an $L_{T}$-formula $\phi(x_1,
  \hdots, x_n)$ such that for all interpretations $m_1, \hdots, m_n$
  of $x_1, \hdots, x_n$ allowed by $T$ and such that $\phi(m_1,
  \hdots, m_n)$ is well-sorted, we have that $P(m_1, \hdots, m_n)$ is
  true iff $\phi(m_1, \hdots, m_n)$ is true.
\end{Definition}

The fact that $\pi(p,x,y)$ is expressible in terms of $numstr(i,s)$ in
the theory $T_{s,n}$ follows immediately from the reduction from $T_p$
to $T_{s,n}$ used to establish the undecidability theorem
in the previous section. We only have to show the reverse direction,
i.e., that $numstr(i,s)$ is expressible in terms of $\pi(p,x,y)$.
\footnote{Note that we do not present a reduction from $T_{s,n}$ to
  $T_p$.  However, we conjecture that one exists, due to the
  possibility of mapping the countably infinite set of string
  constants onto the countably infinite set of natural numbers and
  then constructing string functions and predicates as operators over
  natural numbers.}

\begin{Theorem}
  $numstr(i,s)$ is expressible in terms of $\pi(p,x,y)$ in $T_\pi$.
\end{Theorem}

\begin{proof}
  We represent $numstr(i,s)$ as a formula that asserts the
  non-existence of a witness for one of two kinds of error in the
  conversion. The first kind of error relates to the maximum possible
  value of $i$. Suppose $s$ is a binary string of length $n$.  Then
  $s$ cannot represent a natural number greater than or equal to
  $2^n$.  The second error is a discrepancy between the binary
  representation of $i$ and the binary string $s$.  To check bit $t$
  of the number $i$, decompose $i$ into $h2^{t+1} + x2^{t} + l$ where
  $x$ is the $t$-th bit of $i$ and so $x = 0 \lor x = 1$, and $l$ is
  the numeric representation of bits $t-1$ through 0 and so $l <
  2^{t}$.  Then if $x = 0$ and $s[len(s) - 1 - t] = ``1"$, or if $x =
  1$ and $s[len(s) - 1 - t] = ``0"$, there is an error.  This gives
  us the following sentence:

  \begin{align*}
    numstr(i, s) \iff & \forall n\, p\, t\, h\, p_{h}\, x\, p_{x}\, l\, l_{u}\, s_h\, s_x\, s_l: \\
    & \lnot (len(s) = n \land \pi(p, 1, n) \land i \ge p) \\
    & \land \lnot (\pi(p_{h}, h, t+1) \land \pi(p_{x}, x, t) \\
    & \land i = p_{h} + p_{x} + l \land \pi(l_{u}, 1, t) \land l < l_u \\
    & \land s = s_h \cdot s_x \cdot s_l \land len(s_l) = t \land len(s_x) = 1 \\
    & \land ((x = 0 \land s_x = ``1") \lor (x = 1 \land s_x = ``0")))
  \end{align*}

  We can apply this rule recursively to the input formula, along with similar rules to the ones presented previously,
  to obtain a reduction from $T_{s,n}$ to $T_{\pi}$.
\end{proof}

\section{Result 3: Axiomatization $\Gamma$ of the Language $L$}
\label{sec:soundness}

In this section, we present a consistent
axiomatization $\Gamma$ for the functions and predicates of the
language $L$ (as presented earlier in Section~\ref{sec:prelim}).

\subsection{The Axiomatization $\Gamma$}
\label{sec:axioms}

We introduce an axiom system $\Gamma$ for the language $L$.
For the sake of readability, we choose not to specify the
sorts of various terms if they are clear from context.

\subsubsection{Axioms of Linear Arithmetic over the Natural Numbers}

The following axioms follow from the ones for Presburger
arithmetic. Note that both Presburger arithmetic and the linear
arithmetic as part of $\Gamma$ include only the addition symbol, and
do not have full multiplication. (Multiplication by constants is
simply a short-hand for repeated addition up to a known constant
bound.)

\begin{enumerate}

\item $0 \neq 1$
  
\item $\forall x : \lnot (0 = x + 1)$

\item $\forall x \exists y : x \ne 0 \to : y + 1 = x$ 

\item $\forall x \, y : \lnot (x < y \land y < x + 1)$ 

\item $\forall x\, y : x + y = y + x$ 

\item $\forall x\, y\, z : (x + y = x + z) \to (y = z)$ 

\item $\forall x\, y : x + 1 = y + 1 \to x = y$

\item $\forall x : x + 0 = x$

\item $\forall x\, y : x + (y + 1) = (x + y) + 1$

  
\item $\forall x\, y \exists c : x < y \to \lnot (c = 0) \land x + c = y$

\item $\exists c \forall x \, y: \lnot (c = 0) \land x + c = y \to x < y $

\end{enumerate}

\subsubsection{Axioms of Equality for Strings and Natural Numbers}

It is assumed that the equality predicate for both string and numeric
sorts is reflexive, symmetric, and transitive. In addition, we have
the following axiom recursively defined over string terms. Below we
present the axiom for string constants over the alphabet $\Sigma$.

\begin{enumerate}[resume]

\item $\forall A \, B : A = B \to len(A) = len(B)$

\end{enumerate}

\subsubsection{Axioms of Concatenation}

Concatenation is associative, but not commutative.

\begin{enumerate}[resume]

\item $\forall x : x \cdot \epsilon = \epsilon \cdot x = x$

\item $\forall xyz : x \cdot (y \cdot z) = (x \cdot y) \cdot z$

\end{enumerate}

\subsubsection{Axioms of the $len$ Function}

\begin{enumerate}[resume]

\item $\forall x : len(x) = 0 \iff x = \epsilon$

\item $\forall x : len(x) = 1 \to \bigvee_{c \in \Sigma} x = c$

\item $\forall x\, y : len(x \cdot y) = len(x) + len(y)$

\item $\forall c \in \Sigma : len(c) = 1$

\end{enumerate}
  
\subsubsection{Axioms of $numstr$}

The axioms for the $numstr$ predicate essentially allow us to define a
natural mapping between natural numbers, represented in binary, and
strings over $\Sigma$.

\begin{enumerate}[resume]

\item $\forall i : \lnot numstr(i, \epsilon)$

\item $numstr(0, ``0")$

\item $numstr(1, ``1")$

\item $\forall s\, i : len(s) = 1 \land s \ne ``0" \land s \ne ``1" \to \lnot numstr(i, s)$
  
\item $\forall i\, x\, z : numstr(i, x) \land ``0"z = z``0" \to numstr(i, zx)$

\item $\forall i\, x\, z : numstr(i, zx) \land ``0"z = z``0" \land z \ne \epsilon \land x \ne \epsilon \to numstr(i, x)$


\item $\forall x\, y\, z : (\exists u\, v : numstr(u, y) \land numstr(v, z)) \to (numstr(x, yz) \iff x = u_b v_b)$, where $u_b$ and $v_b$ are the binary
  digits of $u$ and $v$ respectively. (This describes distribution of $numstr$ over a concatenation.)


\item $\forall x\, y\, z \exists u \, v \, w : numstr(x + y, z) \to : len(u) = x \land len(v) = y \land w = uv \land numstr(len(w), z)$

\item $\exists u\, v\, w \forall x \, y \, z: len(u) = x \land len(v) = y \land w = uv \land numstr(len(w), z) \to numstr(x + y, z)$

\end{enumerate}

\subsection{Relationship between $T_{\Gamma}$ and $T_{s,n}$}

We refer to the set of sentences logically entailed by the axiom
system $\Gamma$ as the theory $T_{\Gamma}$. Note that this set
contains sentences with arbitrary quantifiers in them. We assume that
sentences are always written in prenex normal form.

The set $T_{s,n}$ is a set of quantifier-free $L$-formulas. As discussed
before, when we use the term ``quantifier-free'' formulas, we mean
that each free variable is implicitly existentially quantified and
there are no other explicit quantifiers in the formula. When the
formulas in $T_{s,n}$ are existentially quantified, we get the same set
of sentences implied by $\Gamma$ that have a single set of existential
quantifiers in prenex normal form. We also call this the existential
fragment of $T_{\Gamma}$.

\subsection{Consistency of $\Gamma$}

\begin{Theorem}
  The axiom system $\Gamma$ presented in Section~\ref{sec:axioms} is consistent.
\end{Theorem}

\begin{proof}
  It is well known that a theory or axiom system is consistent if it
  has a model~\cite{HodgesModelTheory}.  We prove consistency by
  showing that the structure established in
  Section~\ref{sec:canonicalmodel} is in fact a model of $\Gamma$. The
  remainder of the proof is structured in sections corresponding to
  those in the description of $\Gamma$.

  \begin{enumerate}

  \item \textbf{Axioms of arithmetic over natural numbers:} These are
   standard axioms for natural number arithmetic.  Since we choose
   $\mathbb{N}$ to model numeric terms, it follows that these axioms
   are true over the natural numbers.

 \item \textbf{Axioms of equality for strings and natural numbers:}
   This axiom states that if two strings $A$ and $B$ are equal, then
   $A$ and $B$ have the same length, in addition to the standard
   axioms of equality. Our model of string terms states that two
   strings are equal if they have the same characters appearing in the
   same order, and that the length of a string is the natural number
   of characters in that string.  It follows that if two strings are
   equal, then they have the same characters, and therefore have the
   same length.

 \item \textbf{Axioms of concatenation:} The first axiom states that
   concatenating any string with the empty string, on either side,
   produces a result equal to the original string.  Our model
   represents the result of concatenating $A$ and $B$ as a string
   having all of $A$'s characters (in the same order) followed by all
   of $B$'s characters (also in the same order). If one of $A$ or $B$
   is empty, it follows that the resulting string has the same
   characters and in the same order as the other string, and therefore
   the two are equal.

   The second axiom states that string concatenation is associative.
   Suppose strings $X, Y, Z$ are composed of characters $x_1 \hdots
   x_u$, $y_1 \hdots y_v$, $z_1 \hdots z_w$ respectively. Then by
   definition of concatenation in our model, we have:

  \begin{align*}
    y \cdot z & = y_1 \hdots y_v z_1 \hdots z_w \\
    x \cdot (y \cdot z) & = x_1 \hdots x_u y_1 \hdots y_v z_1 \hdots z_w \\
    x \cdot y & = x_1 \hdots x_u y_1 \hdots y_v \\
    (x \cdot y) \cdot z & = x_1 \hdots x_u y_1 \hdots y_v z_1 \hdots z_w \\
    x \cdot (y \cdot z) & = (x \cdot y) \cdot z
  \end{align*}

  Therefore the axiom holds in this model.

\item \textbf{Axioms of the Length Function:} The first axiom states
  that the only string having length 0 is the empty string $\epsilon$,
  which follows trivially from the definition of the set of string
  constants $\Sigma^{*}$.

  The second axiom states that the length of the concatenation of $A$
  and $B$ is equal to the sum of the lengths of $A$ and $B$ taken
  separately.  Our model represents the result of concatenating $A$
  and $B$ as a string having all of $A$'s characters (in the same
  order) followed by all of $B$'s characters (also in the same
  order). Since characters are conserved by this process, it follows
  that the resulting string has length equal to the sum of the lengths
  of $A$ and $B$.

  The third axiom states that all single-character strings have length
  1, which holds trivially.

\item \textbf{Axioms of $numstr$ string-numeric conversion predicate:}
  The first four axioms state some basic properties of string-number
  conversion: $\epsilon$ is not the binary representation of any
  number, ``0'' is the binary representation of 0, ``1'' is the binary
  representation of 1, and single-character strings that are not ``0''
  or ``1'' are not the binary representation of any number. These
  axioms are true by inspection.

  The fifth and sixth axioms show that leading zeroes can be added to and removed from
  a string without changing its value. 
   We can show that this is true in our model by demonstrating that if $y$ is a binary string and $z$ is a string of all zeroes,
   the binary expansions of $y$ and $zy$, denoted $y_b$ and $(zy)_{b}$ respectively, both represent the same natural number:
  
   \begin{align*}
     y_{b} & = y[0] 2^{length(y)-1} + y[1] 2^{length(y) - 2} + \hdots \\
     & + y[length(y) - 2] 2^{1} + y[length(y)-1] 2^{0} \\
     (zy)_{b} & = (zy)[0] 2^{length(zy) - 1} + (zy)[1] 2^{length(zy) - 2} + \hdots \\
     & + (zy)[length(z)-1] 2^{length(zy) - length(z) - 2} \\
     & + (zy)[length(z)] 2^{length(zy) - length(z) - 1} + \hdots \\
     & + (zy)[length(zy) - 1] 2^{length(zy) - length(zy)} \\
     (zy)[0] & = 0 \\
     (zy)[1] & = 0 \\
     \vdots & \\
     (zy)[length(z)-1] & = 0 \\
     (zy)_{b} & = 0 + 0 + \hdots + 0 \\
     & + (zy)[length(z)] 2^{length(zy) - length(z) - 1} + \hdots \\
     & + (zy)[length(zy) - 1] 2^{length(zy) - length(zy)} \\
     (zy)[length(z)] & = y[0] \\
     (zy)[length(z) + 1] & = y[1] \\
     \vdots & \\
     (zy)[length(z) + length(y) - 1] & = y[length(y) - 1] \\
     length(zy) - length(z) & = length(y) \\
     (zy)_{b} & = y[0] 2^{length(y) - 1} + \hdots + y[length(y) - 1] 2^{0} \\
     & = y_{b}
   \end{align*}
  
  Hence adding or deleting leading zeroes has no effect on what number is represented by a given binary string, and so these axioms hold.
  
  The seventh axiom holds if we assume that all numbers are written in binary;
  concatenating the binary digits of two numbers is equivalent to
  concatenating the string representations of those numbers.
  
%
%

  The eighth axiom illustrates how to perform string-number conversion
  on an addition term $x + y$.  It suffices to show
  that $len(w) = x + y$:

  \begin{align*}
    w & = uv \\
    len(w) & = len(u) + len(v) \\
    & = x + y
  \end{align*}

  \end{enumerate}
  
This completes the proof.
\end{proof}

\section{Result 4: Incompleteness of the Theory $T_{\Gamma}$}
\label{sec:incompleteness}

We first state a number of useful definitions and theorems related to
completeness of first-order theories from the standard model theory
literature~\cite{HodgesModelTheory}.

\begin{Definition}
A first-order theory $T$ in language $L$ is \textbf{complete} if for
all $L$-formulas $\phi$, exactly one of $\phi$ and $\lnot \phi$ is a
consequence of $T$.
\end{Definition}

\begin{Definition}
Two models $A, B$ of a first-order theory are \textbf{elementarily
  equivalent} if for all first order $L$-sentences $\phi$, $A \vDash
\phi \iff B \vDash \phi$.
\end{Definition}

\begin{Theorem}
\label{thm:completeIffEquivalent}
A first-order theory $T$ is complete iff all of its models
are elementarily equivalent~\cite{HodgesModelTheory}.
\end{Theorem}

\noindent{We are now in a position to prove the following result.}

\begin{Theorem}
$T_{\Gamma}$ is incomplete.
\end{Theorem}

\begin{proof}
  Consider two models $A, B$ of the theory $T_{\Gamma}$, defined as
  follows: $A$ is the canonical model given in
  Section~\ref{sec:canonicalmodel}, and $B$ is a restricted version of
  $A$ where the only string constants that are allowed
  are nonempty string constants with no leading zeroes. (In other
  words, the only string constant in $B$ that starts with `0' is
  ``0''.) It is easy to see that both of these are models of
  $T_{\Gamma}$.

  Now consider the first-order sentence $J$ which states ``the
  $numstr$ predicate describes a bijection between strings and natural
  numbers''.
\footnote{Note that as long as the alphabet $\Sigma$ is finite and
  string constants are concatenations of a finite number of
  characters, in general there exists a bijection between strings and
  natural numbers. This follows from the fact that the set $\Sigma^*$
  of strings is countably infinite. The argument made in the proof
  above deals with a very particular bijection as defined by
  $numstr$.}

We state this sentence $J$ formally as follows:

\begin{align*}
& \forall_{num} \, i : \exists_{str} \, s : \left( numstr(i,s) \land \forall_{str} \, t : numstr(i,t) \to s = t \right) \\
\land & \forall_{str} \, s : \exists_{num} \, i : \left( numstr(i,s) \land \forall_{num} \, j : numstr(j,s) \to i = j \right)
\end{align*}

  It follows that due to the restrictions on string constants imposed
  in $B$, $numstr$ clearly defines a bijection between strings and
  natural numbers, where each integer is mapped to the unique string
  that is its minimal binary representation, and so $J$ is true in the
  model $B$.  However, in the model $A$, $numstr$ does not define a
  bijection, as by counterexample, $numstr(3, ``11")$ and $numstr(3,
  ``0011")$ are both true.  Therefore $J$ is false in the model $A$.
  
  From this we conclude that $J$ is able to distinguish between $A$
  and $B$, and hence $A$ is not elementarily equivalent to $B$; by
  Theorem~\ref{thm:completeIffEquivalent}, $T_{\Gamma}$ is
  incomplete.
\end{proof}

%
%

\section{Related Work}
\label{sec:relwork}

We provide a relatively comprehensive overview of both theoretical and
practical work done by researchers in the context of theories over
strings.

\subsection{Theoretical Results over Theories of Strings}

In his original 1946 paper, Quine \cite{Quine} showed that the
first-order theory of string equations (i.e., quantified sentences
over Boolean combination of word equations) is undecidable.  Due to
the expressibility of many key reliability and verification questions
within this theory, this work has been extended in many ways.

One line of research studies fragments and modifications of this base
theory which are decidable.  Notably, in 1977, Makanin proved that the
satisfiability problem for the quantifier-free theory of word
equations is decidable \cite{makanin}.  In a sequence of papers,
Plandowski and co-authors showed that the complexity of this problem
is in PSPACE \cite{plandowski2006}.  Stronger results have been found
where equations are restricted to those where each variable occurs at
most twice\cite{robsondiekert} or in which there are at most two
variables \cite{CharaPach, IliePland, dabrowski2002weo}. In the first
case, satisfiability is shown to be NP-hard; in the second, polynomial
(which was improved further in the case of single variable word
equations). Concurrently, many researchers have looked for the exact
boundary between decidability and undecidability. Durnev \cite{durnev}
and Marchenkov \cite{marchenkov} both showed that $\forall\exists$
sentences over word equations is undecidable.  Despite decades of
effort, however, the satisfiability problem for the quantifier-free
theory of word equations and numeric length remains
open~\cite{makanin,plandowski2006,ganesh2012,Matiyasevich}.  More
recently, Artur J\"{e}z presents a technique called recompression that
gives more efficient algorithms for many fragments of theory of word
equations~\cite{jez}.

A related result was shown by
Furia~\cite{DBLP:journals/corr/abs-1001-2100}, wherein he proved that
the quantifier-free theory of integer sequences is decidable.  The
framework he establishes in that paper is closely related to the
theory of concatenation and word equations, but weaker than either
strings plus numeric length or the theory of arrays due to the
inability of the theory of sequences
to express facts relating indices directly to elements.

Word equations augmented with additional predicates yield richer
structures which are relevant to many applications, as we have
considered here.  In the 1970s, Matiyasevich formulated a connection
between string equations augmented with integer coefficients whose
integers are taken from the Fibonacci sequence and Diophantine
equations~\cite{MatiyasevichHilbertPub,Matiyasevich}.  In particular, he showed that proving
undecidability for the satisfiability problem of this theory would
suffice to solve Hilbert's Tenth Problem in a novel way.

Schulz \cite{schulz} extended Makanin's satisfiability algorithm to
the class of formulas where each variable in the equations is
specified to lie in a given regular set (i.e. a set defined by a regular language).
This is a strict generalization of the solution sets of word equations.  Further work
in~\cite{KarhumakiPM97} shows that the class of sets expressible
through word equations is incomparable to that of regular sets.
Matiyasevich extends Schulz's result to decision problems
involving trace monoids and free partially commutative monoids~\cite{matiyasevichExtra1,matiyasevichExtra2,matiyasevichExtra3}.

M\"oller~\cite{Moller} studies word equations and related theories as
motivated by questions from hardware verification. More specifically,
M\"oller proves the undecidability of the existential fragment of a
theory of fixed-length bit-vectors, with a special finite but parameterized
concatenation operation, extraction of substrings, and
equality predicate. Although this theory is related to the word
equations that we study, it is more powerful because of the
finite but possibly arbitrary concatenation.

The question of whether the satisfiability problem for the
quantifier-free theory of word equations and length constraints is
decidable has remained open for several decades.  Our decidability
results are a partial and conditional
solution. Matiyasevich~\cite{matiyasevich2008} observed the relevance
of this question to a novel resolution of Hilbert's Tenth Problem.  In
particular, he showed that if the satisfiability problem for the
quantifier-free theory of word equations and length constraints is
undecidable, then it gives us a new way to prove Matiyasevich's
Theorem (which resolved the famous problem)~\cite{Matiyasevich,
  matiyasevich2008}.

B\"{u}chi et al.~\cite{Buchi90} consider extensions of the
quantifier-free theory of word equations with various length
predicates. They find that a predicate $Elg$ that asserts that two
strings have equal length is not existentially definable in this
theory, and that by introducing two stronger functions,
$Lg_{1}$ and $Lg_{2}$ which count the number of occurrences of the
characters `1' and `2' respectively, the resulting theory is
undecidable.

The source of undecidability, as the authors identify, is the
ability for these functions to match the number of occurrences of
certain subsequences, which allows them to encode addition and
multiplication.  Our result is similar to
this one; B\"{u}chi proposes an encoding of arithmetic into word
equations, while we assume an extension of word equations that already
contains the $len$ function and natural number arithmetic (as well as
$numstr$), and encode an arithmetic operation into operations on
strings.

\subsection{String Solvers and their application in Program Analysis, Bug-finding, and Verification}

Formulas over strings became important in the context of automated
bug-finding~\cite{hampi, prateek} and analysis of database/web
applications~\cite{emmiMS2007, rupak, WassermannSu2007}. These program
analysis and bug-finding tools read string-manipulating programs and
generate formulas expressing their outputs.  These formulas contain
equations over string constants and variables, membership queries over
regular expressions, inequalities between string lengths, and in some
cases the string-integer conversion predicate/functions.  In practice,
formulas of this form have been solved by off-the-shelf solvers such
as HAMPI~\cite{hampi2, hampi}, Z3str2~\cite{z3strcav},
CVC4~\cite{Liang2014}, or Kaluza~\cite{prateek}. All these solvers are
based on sound algorithms, but are incomplete in different ways.

Zheng et al. \cite{z3strcav} present the Z3str2 solver for the
quantifier-free many-sorted theory $T_{wlr}$ over word equations,
membership predicate over regular expressions, and length function,
which consists of the string (str) and numeric (num) sorts.

S3~\cite{s3} is another tool that supports word equations, length
function, and regular expression membership predicate. S3 internally
uses a version of Z3str2 to handle word equations and length
functions.

CVC4 \cite{Liang2014} handles constraints over the theory of unbounded
strings with length and RE membership. Solving is based on
multi-theory reasoning backed by the DPLL($T$) architecture combined
with existing SMT theories.  The Kleene star operator in RE formulas
is dealt with via unrolling as in Z3str2.

In a separate paper, Liang et al.~\cite{Liang2015} give a decision
procedure for regular language membership and numeric length
constraints over unbounded strings. However, their decision procedure
does not consider word equations, and hence is many ways weaker than
the theory $T_{s,n}$ we consider in this paper.  Hence the algorithm
they propose, while useful in some contexts, is weaker than the full
theory of strings, and their result does not yet resolve the question
of whether the quantifier-free theory of strings and numeric length
constraints is decidable.

It must be stressed that all the solvers (including Z3str2, CVC4, and
S3) that purportedly solve the satisfiability problem for the theory
$T_{s,n}$ or the word equation and length function fragment of
$T_{s,n}$ are incomplete. Solvers such as HAMPI are limited by the
fact that they reason only over a bounded string domain, where the
bound is given as part of the input.

Pex \cite{Tillmann2008} is a parameterized unit testing tool for .NET
that observes program behaviour and uses a constraint solver in order
to produce test inputs which exercise new program behaviour.  It
integrates a specialized string solver in order to generate string
inputs that satisfy the desired branch conditions.

\section{Conclusions and Future Work}
\label{sec:conclusions}

In recent years there has been considerable interest in satisfiability
procedures (aka solvers) for theories over string equations, length,
and string-number conversions in the verification and security
communities~\cite{z3strcav, Liang2014}. These theories are also of
great interest to logicians, since there are many open problems
related to their decidability and complexity. We showed that a
first-order many-sorted quantifier-free theory $T_{s,n}$ of string
equations, linear arithmetic over length function, and string-number
conversion predicates, variations of which have been implemented in
solvers such as Z3str2 and CVC4, is undecidable. We establish expressibility
results for $numstr$ predicate that suggest that this predicate is far
more complex than appears at first glance. Finally, we also provide a
consistent axiomatization $\Gamma$ for the symbols of $T_{s,n}$, and
show that the theory $T_{\Gamma}$ is incomplete.

There are many decidability, complexity and efficient encoding
questions related to fragments of $T_{s,n}$ that remain open. For
example, it is not known whether the theory of word equations and
arithmetic over length functions is decidable~\cite{Matiyasevich}. The
satisfiability problem for the quantifier-free theory of string
equations by itself is known to be in PSPACE; however, it is not known
whether it is PSPACE-complete~\cite{plandowski2006}. Yet another open
question concerns efficient encoding of functions such as ``Replace''
that are heavily used in many programming languages, and predicates
such as string comparison. More generally, efficient encoding of
common programming language string-intensive functions and predicates
in terms of $T_{s,n}$-functions and predicates can be of great value
to practitioners, and remains a challenging practical problem.

\bibliographystyle{abbrv}
\bibliography{strings}
\end{document}